\documentclass[showpacs,pra,twocolumn,superscriptaddress]{revtex4-1}

\usepackage{exscale}
\usepackage{bbm}
\usepackage{graphicx}
\usepackage{latexsym}
\usepackage{times}
\usepackage[T1]{fontenc}
\usepackage{enumerate}
\usepackage{bbold}
\usepackage{color}
\usepackage{mathtools}
\usepackage{changes}
\usepackage[normalem]{ulem}
\usepackage{amsfonts,amsmath,amssymb,amsthm,dsfont}
\usepackage[colorlinks=true,citecolor=blue,urlcolor=blue]{hyperref}

\usepackage{physics}

\theoremstyle{plain}
\newtheorem{thm}{Theorem}
\newtheorem{lem}{Lemma}
\newtheorem{fact}{Fact}

\begin{document}
\title{All genuinely entangled stabilizer subspaces are multipartite fully nonlocal}
\author{Owidiusz Makuta}
\affiliation{Center for Theoretical Physics, Polish Academy of Sciences, Aleja Lotnik\'{o}w 32/46, 02-668 Warsaw, Poland}

\author{Remigiusz Augusiak}
\affiliation{Center for Theoretical Physics, Polish Academy of Sciences, Aleja Lotnik\'{o}w 32/46, 02-668 Warsaw, Poland}

\begin{abstract}
Understanding which entangled states give rise to Bell nonlocality and thus are resourceful in the device-independent framework is a long-stanging unresolved problem. Here we establish the equivalence between genuine entanglement and genuine nonlocality for a broad class of multipartite (pure and mixed) states originating from the stabilizer formalism. In fact, we prove that any (mixed) stabilizer state defined on a genuinely entangled subspace is multipartite fully nonlocal meaning that it gives rise to correlations with no contribution from local hidden variable models of any type. Importantly, we also derive a lower bound on genuine nonlocality content of arbitrary multipartite states, opening the door to its experimental estimation. 
\end{abstract}

\maketitle

\textit{Introduction.}
Quantum entanglement and Bell nonlocality are some of the most characteristic features of quantum theory. Moreover, they are a cornerstone of quantum information science enabling several applications unachievable in classical physics such as quantum teleportation \cite{Werner_2001}, quantum cryptography \cite{BB84,E91} or certification of random numbers \cite{di4}. In some of these applications such as the latter two, Bell nonlocality proves to be a stronger resource as it allows to process information in the device-independent framework \cite{PhysRevLett.98.230501,di4} where no assumptions are made on the internal working on the devices except that they obey the rules of quantum theory.

It is well-known that both these resources are deeply related as Bell nonlocality can only be obtained from entangled states; yet they are not fully equivalent. Indeed, already in the bipartite case there exist mixed entangled states which do not violate any Bell inequality such as e.g. the well-known Werner states \cite{PhysRevA.40.4277,PhysRevA.65.042302}, despite the fact that in the case of pure states, entanglement implies Bell non-locality \cite{GISIN1991201,GISIN199215}. In fact, understanding which entangled states can give rise to Bell violations and thus be a resource in the device-independent framework remains an unsolved and highly nontrivial problem. 

While being already difficult to tackle in the bipartite case, the question which entangled states are nonlocal gets even more complicated when one enters the realm of multipartite quantum systems which are far more complex to handle.
It was proved in Refs. \cite{POPESCU1992293, Gachechiladze_2017} as a generalization of the results of \cite{GISIN1991201,GISIN199215} that all pure multipartite entangled states exhibit some form of nonlocality. However, an arguably more meaningful question in the multipartite scenario, i.e., whether all genuinely multipartite entangled (GME) states are genuinely multipartite nonlocal (GMNL) \cite{PhysRevLett.115.030404} remains unresolved even for the pure states. While nothing is known in the most general case, for some particular multipartite systems such as pure three-qubit \cite{YuOh} or $N$-qubit symmetric \cite{PhysRevLett.112.140404} states this equivalence was established. At the same time, there exist mixed GME states that are not GMNL \cite{PhysRevLett.115.030404}. In fact, it is even possible to construct GME mixed states that are fully local and thus do not exhibit any form of nonlocality \cite{Geneva}. 

Here we establish the equivalence between genuine multipartite entanglement and nonlocality for a large class of pure and mixed states originating from the stabilizer formalism \cite{Gottesman}. The latter provides a very convenient representation of multipartite states that encompass not only the well-known graph states, but also a broad class of genuinely entangled mixed states. This includes for instance mixed states corresponding to stabilizer quantum error correction codes \cite{Gottesman} such as the five-qubit \cite{PhysRevLett.77.198} or toric ones \cite{Kitaev1997}. In this work, we exploit this representation to show that every genuinely entangled stabilizer subspace of the $N$-qubit Hilbert space is also genuinely nonlocal in the sense that every pure state belonging to it gives rise to genuinely nonlocal correlations. We thus introduce the first examples of multi-qubit Hilbert spaces composed of only genuine multipartite nonlocal pure states; see Refs. \cite{Demianowicz,Demianowicz2022universal} for constructions of genuinely entangled subspaces.
We actually prove a much stronger result that any state (pure or mixed) belonging to a genuinely entangled stabilizer subspace is multipartite fully nonlocal (MFNL) which means that it gives rise to nonlocal correlations that are genuinely nonlocal in the strongest sense, i.e., they have no contribution coming from a local hidden variable model of any type (see Refs. \cite{Gallego1,Gallego2,Cabello20230,Cabello2023} for previous examples of fully nonlocal correlations).
Importantly, this implication generalizes to mixed stabilizer states: any mixed state defined on a GME stabilizer subspace is also multipartite full nonlocal. We thus introduce here a broad class of mixed states that give rise to multipartite fully nonlocal correlations. 
Lastly, we derive a general lower bound on the genuine nonlocality content of any multipartite state or subspace (not necessarily stabilizer), based on the nonlocality contents of bipartite states obtained from local measurements performed by the remaining parties on the multipartite state and subspace respectively.

\textit{Preliminaries.} Let us first provide some background information.

\textit{Stabilizer formalism.} We begin with the stabilizer formalism and define the $N$-qubit Pauli group $\mathbb{P}_N$ to be one containing all $N$-fold tensor product of Pauli matrices $\mathbb{1}, X,Y,Z$ multiplied by $\pm 1$ or $\pm \mathbb{i}$.
Let us then consider a subgroup $\mathbb{S}$ of $\mathbb{P}_{N}$. We call it a stabilizer if it is abelian and satisfies $-\mathbb{1}\notin\mathbb{S}$.
The most interesting property of $\mathbb{S}$ is that, as the name suggests, its elements stabilize a non-empty subspace in the $N$-qubit Hilbert space $\mathcal{H}_N=(\mathbb{C}^2)^{\otimes N}$. Precisely, for any stabilizer $\mathbb{S}$ there exists a subspace $V \subseteq \mathcal{H}_N$ such that $s\ket{\psi}=\ket{\psi}$ for any $\ket{\psi}\in V$ and any $s\in\mathbb{S}$. The largest subspace satisfying this condition is called a stabilizer subspace of $\mathbb{S}$ and is denoted $V_{\mathbb{S}}$. Clearly, this stabilizing property extends to all mixed states acting on a subspace $V_{\mathbb{S}}$, that is, for any $\rho:V_{\mathbb{S}} \to V_{\mathbb{S}}$ one has that $s\rho=\rho s=\rho$ for any $s\in\mathbb{S}$. 
Thus, the most basic function of a stabilizer $\mathbb{S}$, which we extensively use here, is to uniquely define a subspace in terms of a few algebraic relations that are convenient to handle. 

Here we often need to consider matrices forming operators $s\in\mathcal{S}$ that correspond only to subsets of $[N]=\{1,\dots,N\}$. Hence, we denote by $s^{(Q)}$ a $|Q|$-fold tensor product of the Pauli matrices from $s$ that act on qubits belonging to the set $Q\subset [N]$. For instance, for $s=X\otimes Z\otimes \mathbb{1}\otimes XZ$, the sub-operator $s^{(Q)}$ corresponding to $Q=\{1,3\}$ is $s^{(Q)}=X\otimes \mathbb{1}$.

Lastly, since for larger $N$ the cardinality of $\mathbb{S}$ could substantially increase, it is convenient to represent $\mathbb{S}$ in terms of the minimal set of operators, called generators, that enable reproducing the other elements of $\mathbb{S}$. By writing $\mathbb{S}=\langle g_{1},\dots,g_{k}\rangle$ we mean that $\{g_{i}\}_{i=1}^{k}$ is a generating set of $\mathbb{S}$. For instance, the stabilizer $\mathbb{S}=\{\mathbb{1}\otimes\mathbb{1},X\otimes X, Z\otimes Z, -Y\otimes Y\}$ can be represented in this way as $\mathbb{S}=\langle X\otimes X, Z\otimes Z\rangle$.

\textit{Genuine multipartite entanglement.} We can now move on to the definition of genuine entanglement in the multipartite setting. We consider an arbitrary $N$-partite Hilbert space $\mathcal{H}=\mathcal{H}_{1}\otimes\dots\otimes \mathcal{H}_{N}$ and denote by $\mathcal{B}(\mathcal{H})$ the operator space over $\mathcal{H}$. Let us then divide the set $[N]\equiv \{1,\ldots,N\}$ into two disjoint and non-empty sets $Q$ and $\overline{Q}$ and call it a bipartition $Q|\overline{Q}$. A state $\ket{\psi}\in\mathcal{H}$ is called genuinely multipartite entangled if for all bipartitions, $\ket{\psi}\neq \ket{\psi_Q}\otimes|\phi_{\overline{Q}}\rangle$
for any two pure states $\ket{\psi_Q}$ and $|\phi_{\overline{Q}}\rangle$ corresponding to the sets $Q$ and $\overline{Q}$. 

Moving to the mixed-state case, we say that a mixed state $\rho\in\mathcal{B}(\mathcal{H})$ is genuinely multipartite entangled (GME) \cite{PhysRevA.65.012107} if it does not admit a biseparable model, meaning that it cannot be written 
as a convex combination,
\begin{equation}
\rho = \sum_{Q\subset[N]}p_Q \sum_{i}q_{i;Q}\,\rho_Q^{i}\otimes \rho_{\overline{Q}}^i,
\end{equation}
of states which are separable across various bipartitions $Q|\overline{Q}$.
In an analogous manner, one can also define genuine entanglement for subspaces: we say that a subspace $V\subset\mathcal{H}$ is GME if all pure states belonging to it are GME \cite{Demianowicz}. Let us stress here that for a genuinely entangled subspace $V$, every mixed state defined on it is genuinely entangled too.

Our focus in this work is the stabilizer subspaces, and therefore we need a criterion that allows us to easily decide whether a given stabilizer subspace $V_{\mathbb{S}}$ is GME. Such a criterion was recently introduced in \cite{Makuta_2021} (see Theorem 1 therein), and for further purposes we recall it here as the following lemma.
\begin{lem}\label{lem:gme}
Given a stabilizer $\mathbb{S}=\langle g_{1},\dots, g_{k}\rangle$, the corresponding stabilizer subspace $V_{\mathbb{S}}$ is GME iff for each bipartition $Q|\overline{Q}$ there exists a pair $i,j\in [k]$ for which:
\begin{equation}\label{eq:lem_ncom_org}
\left[g_{i}^{(Q)},g_{j}^{(Q)}\right]\neq 0.
\end{equation}
\end{lem}

\textit{Genuine multipartite nonlocality.}
Let us now consider a typical Bell scenario with $N$ parties sharing a state $\rho\in\mathcal{B}(\mathcal{H})$ and performing measurements on their shares of this state. Party $i$ can freely choose to perform one of $x_i=1,\ldots,m$ measurements which yields an outcome $a_i=0,\ldots,d-1$. 
By repeating these measurements many times, the parties create correlations that are described by a collection of joint probabilities $\mathcal{P}=\{P(\mathbf{a}|\mathbf{x})\}$, often referred to as behavior, where $P(\mathbf{a}|\mathbf{x})$ is a probability of obtaining outcomes $a_1,\ldots,a_N=:\mathbf{a}$ after performing measurements $x_1,\ldots,x_N=:\mathbf{x}$. 

In order to formulate the notion of genuine multipartite nonlocality we need the concept of non-signaling behaviors. Let us for a moment abstract from quantum behaviors and introduce a broader class of behaviors for which a measurement choice made by parties belonging to a subset $Q$ does not have an influence over the measurement result of the remaining parties. Formally, this is represented by a set of linear constraints,
\begin{equation}\label{eq:no-signaling}
P(\mathbf{a}_{\overline{Q}}|\mathbf{x}_{\overline{Q}})=\sum_{i\in Q}\sum_{a_{i}} P(\mathbf{a}|\mathbf{x})
\end{equation}
for all $Q\subset [N]$, all $\mathbf{a}_{Q}$, and all $\mathbf{x}$, where $\mathbf{a}_{Q}=\{a_{i}\}_{i\in Q}$, $\mathbf{x}_{Q}=\{x_{i}\}_{i\in Q}$, and the sum $\sum_{a_{i}}$ is over all possible measurement results. We call behaviors satisfying the above constraints non-signaling. 
Note that all behaviors originating from quantum theory are non-signaling;
yet there exist non-signaling correlations that are not quantum.

Let us now consider a behavior $\mathcal{P}$. We call it 
local with respect to a given bipartition $Q|\overline{Q}$ if 
\begin{equation}\label{eq:locality}
P(\mathbf{a}|\mathbf{x}) = \sum_{\lambda} q_{\lambda} P(\mathbf{a}_{Q}|\mathbf{x}_{Q},\lambda)P(\mathbf{a}_{\overline{Q}}|\mathbf{x}_{\overline{Q}},\lambda)
\end{equation}
for all $P(\mathbf{a}|\mathbf{x})\in\mathcal{P}$,
where $\lambda$ is the hidden variable with a distribution $q_{\lambda}$, and 
$P(\mathbf{a}_{Q}|\mathbf{x}_{Q},\lambda)$ and 
$P(\mathbf{a}_{\overline{Q}}|\mathbf{x}_{\overline{Q}},\lambda)$ are some
non-signaling (in general non-quantum) behaviors corresponding to the disjoint sets $Q$ and $\overline{Q}$ that satisfy the non-signaling conditions. Similarly to the case of entanglement, we call $\mathcal{P}$ genuinely multipartite non-local (GMNL) \cite{PhysRevD.35.3066, PhysRevA.88.014102} if it cannot be written as a convex combination of behaviors that are local with respect to various bipartitions.

While the above notion allows us to describe nonlocality in a quantitative way, it does not tell us much about how strong this nonlocality is. A possible approach to quantify GMNL in $\mathcal{P}$, put forward in Ref. \cite{Almeida_2010} as a multipartite generalization of the Elitzur-Popescu-Rohrlich (EPR-2) decomposition \cite{EPR2}, is through the following convex decomposition 
\begin{equation}\label{eq:decomposition}
P(\mathbf{a}|\mathbf{x})=\sum_{Q\in[N]}p_{Q|\overline{Q}} P_{Q|\overline{Q}}(\mathbf{a}|\mathbf{x})+p_{NL}P_{NL}(\mathbf{a}|\mathbf{x}),
\end{equation}
where $P_{Q|\overline{Q}}(\mathbf{a}|\mathbf{x})$ is a behavior that admits the decomposition (\ref{eq:locality}) for a given bipartition $Q|\overline{Q}$, whereas $P_{NL}(\mathbf{a}|\mathbf{x})$ is one that is nonlocal across any bipartition, and, finally, $p_{Q|\overline{Q}}$ for all $Q$ and $p_{NL}$ form a probability distribution; in particular $\sum_{Q\in[N]}p_{Q|\overline{Q}}+p_{NL}=1$.

Now, the minimal $p_{NL}$ for which (\ref{eq:decomposition}) holds true, denoted $\tilde{p}_{NL}$, is called \textit{genuine entanglement content} of $\mathcal{P}$. A behavior $\mathcal{P}$ for which $\tilde{p}_{NL}>0$ is 
genuinely entangled, and in the extreme case of $\tilde{p}_{NL}=1$ 
we call it multipartite fully non-local \cite{Almeida_2010}. Thus, correlations that are MFNL
are also GMNL, meaning that the former is a stronger form of 
multipartite nonlocality than the latter. 

Notably, both these notions can also be applied to quantum states: we say that a state $\rho$ is MFNL or GMNL if there exists a Bell scenario and a set of measurements such that the resulting behavior $\mathcal{P}$ is MFNL or GMNL, respectively.

Lastly, in order to detect MFNL (and thus also GMNL) of a given $\rho$, we will make use of Theorems 1 and 2 from \cite{Almeida_2010} which we below combine into a single lemma.
\begin{lem}\label{lem:MFNL}
A state $\rho$ is said to be MFNL if for every bipartition $Q|\overline{Q}$ it is possible to create a maximally entangled state
\begin{equation}\label{eq:max}
\ket{\phi_{+}}_{i,j}=\frac{1}{\sqrt{d}}\sum_{l=0}^{q-1}\ket{ll}_{i,j}
\end{equation}
for some positive integer $q\geqslant 2$, between $i\in Q$ and $j\in \overline{Q}$ for all possible outcomes of local measurements on the parties from $[N]\setminus\{i,j\}$.
\end{lem}

\textit{Main result.} Here we present our main result that all genuinely entangled stabilizer subspaces are multipartite fully nonlocal. We thus establish, in particular, the equivalence between GME and GMNL for 
a large class of stabilizer states.

Our strategy to prove this statement is quite simple: we aim to show that for any genuinely entangled stabilizer subspace, it follows from Lemma \ref{lem:gme} that the conditions of Lemma \ref{lem:MFNL}
are satisfied, that is, that one can create a maximally entangled state across any nontrivial bipartition $Q|\bar{Q}$ by local measurements on $N-2$ parties. 
To achieve this we consider an interesting property of genuinely entangled stabilizer subspaces.
\begin{lem}\label{lem:two_qubits}
Let $\mathbb{S}$ be a stabilizer. The corresponding subspace $V_{\mathbb{S}}$ is GME iff for every pair of qubits $\alpha_{1},\alpha_{2}\in [N]$ there exists a pair of stabilizing operators $s_{i},s_{j}\in\mathbb{S}$ such that
\begin{equation}\label{eq:lem_ncom}
\left[s_{i}^{(\alpha_{l})},s_{j}^{(\alpha_{l})}\right]\neq0, \quad \left[s_{i}^{(\alpha)},s_{j}^{(\alpha)}\right]=0
\end{equation}
for all $\alpha \in[N]\setminus\{\alpha_{1},\alpha_{2}\}$ and all $l\in \{1,2\}$.
\end{lem}

A proof of this lemma can be found in Appendix \ref{app:two_qubits} It is easy to notice that Lemmas \ref{lem:gme} and \ref{lem:two_qubits} are quite similar; in fact, the implication "$\Leftarrow$" in Lemma \ref{lem:two_qubits} follows directly from Lemma \ref{lem:gme}. One important difference, however, is that Lemma \ref{lem:two_qubits} involves all operators from $\mathbb{S}$, whereas Lemma \ref{lem:gme} is formulated only in terms of generators of $\mathbb{S}$.
Still, the main advantage of Lemma \ref{lem:two_qubits} over Lemma 
\ref{lem:gme} is that it tells us that for any pair of qubits
one can always find two operators in $\mathbb{S}$ in which
the local Pauli matrices anticommute exactly for those qubits, whereas 
they commute for the remaining ones.

Utilizing Lemma \ref{lem:two_qubits} we can now formulate our main result.
\begin{thm}\label{thm:mfnl}
A stabilizer subspace $V_{\mathbb{S}}$ is multipartite fully nonlocal iff it is genuinely multipartite entangled.
\end{thm}
Since the proof is quite technical, we have deferred it to Appendix \ref{app:gme_gmnl}. Instead, here we provide a simple example 
illustrating its main idea. We consider the two-dimensional stabilizer subspace, denoted $V_5$, corresponding to the five-qubit code \cite{Bennett_1996, PhysRevLett.77.198}, i.e., one stabilized by $\mathbb{S}_{5}=\langle g_{1},g_{2},g_{3},g_{4}\rangle$, where
\begin{align}
\begin{split}
g_{1}&= X_{1} Z_{2} Z_{3} X_{4},\quad g_{2}=X_{2} Z_{3} Z_{4} X_{5},\\
g_{3}&= X_{1} X_{3} Z_{4} Z_{5},\quad g_{4}= Z_{1} X_{2} X_{4} Z_{5}
\end{split}
\end{align}
with $X_{i},Z_{i}$ denoting the Pauli matrices that act on the party $i$.

Let us then consider a bipartition $\{1,2\}|\{3,4,5\}$. The conditions of Lemma \ref{lem:two_qubits} for this bipartition are fulfilled with $\alpha_{1}=1$, $\alpha_{2}=4$, and $i=3$ and $j=4$, that is, the Pauli matrices at sites $1$ and $4$ of the stabilizing operators $g_3$ and $g_4$ anticommute, whereas 
those appearing at the remaining sites commute. This fact can be used to construct a maximally entangled state between parties $1$ and $4$ from any mixed state $\rho$ defined on $V_5$ by performing suitable measurements on sites $2$, $3$, and $5$. These measurements 
can be constructed from the joint eigenbases of the commuting Pauli operators
appearing at those sites. For instance, for party $2$ we have $g_{3}^{(2)}=\mathbb{1}$ and $g_{4}^{(2)}=X$, and so we pick the measurement for this party in the common eigenbasis of these operators, i.e., $\{\ket{+},\ket{-}\}$, where $\ket{\pm}=(1/\sqrt{2})(\ket{0}\pm\ket{1})$. Analogously, for parties $3$ 
and $5$ we take the measurement bases $\{\ket{+},\ket{-}\}$ and $\{\ket{0},\ket{1}\}$, respectively.

Let us then consider a mixed state $\rho$ defined on $V_5$ and assume that after performing the above measurements on it, the parties observe outcomes corresponding to $\ket{+}$ state for qubits $2$ and $3$, and $\ket{0}$ for qubit $5$. Then, the post-measurement state corresponding to parties $1$ and $4$ reads
\begin{equation}
\sigma_{1,4}=\frac{1}{n}\operatorname{Tr}_{2,3,5}\left[\ket{++0}_{2,3,5}\!\bra{++0}\rho \right],
\end{equation}
where $n$ is the normalization constant. Using the stabilizing relation $\rho=g_{3}\rho$, we further obtain
\begin{eqnarray*}\label{eq:measurements}
\sigma_{1,4}&\!\!=\!\!&\frac{1}{n}\operatorname{Tr}_{2,3,5}[\ket{++0}_{2,3,5}\!\bra{++0} X_{1}X_{3}Z_{4}Z_{5}\rho]\nonumber\\
&\!\!=\!\!&X_{1}Z_{4}\frac{1}{n}\operatorname{Tr}_{2,3,5}[\ket{++0}_{2,3,5}\!\bra{++0} \rho]=X_{1}Z_{4}\sigma_{1,4}.\nonumber
\end{eqnarray*}
In a similar way, one can use the other stabilizing relation $\rho=g_{4}\rho$ to show that $\sigma_{1,4}=Z_1X_4\sigma_{1,4}$. These two relations directly imply that, up to a local unitary, $\sigma_{1,4} = Z_{1}X_{4}\sigma_{1,4}$ is the maximally entangled state of two qubits $\ket{\phi_+}=(\ket{00}+\ket{11})/\sqrt{2}$; recall that the latter is stabilized by 
$X_{1}X_{4}$, $Z_{1}Z_{4}$.

One can verify that for any other choice of the outcomes obtained by parties $2$, $3$, and $5$, the post-measurement state $\sigma_{1,4}$ is stabilized by the same (up to the sign) operators, and thus, it is also the two-qubit maximally entangled state. Furthermore, the same algorithm can be repeated for all bipartitions, and therefore the conditions of Lemma \ref{lem:MFNL} are satisfied in this case. Consequently, the stabilizer subspace corresponding to the five-qubit code is multipartite fully nonlocal.

This theorem thus establishes an equivalence between genuine entanglement and nonlocality in the stabilizer formalism, generalizing results of Ref. \cite{Almeida_2010} derived for qubit graph states. Moreover, it also provides a convenient tool for constructing broad classes of mixed MFNL states for $N\geqslant 4$, as in fact any mixed state supported on a GME stabilizer subspace is MFNL. The above constraint follows from the fact that GME stabilizer subspaces of the smallest nontrivial dimension $\dim(V_{\mathbb{S}})\geqslant 2$ exist in systems consisting of at least four qubits [cf. Theorem 3 in Ref. \cite{Makuta_2021}]. Simultaneously, it is worth pointing out that in the qubit case, the maximal dimension of a GME stabilizer subspace for a given $N$ is $2^{N-k(N)}$, where $k(N)=\lceil (1+\sqrt{8N-7})/2\rceil$ \cite{Makuta_2021}. Thus, Theorem \ref{thm:mfnl} establishes that GME and GMNL are equivalent notions for a large class of multipartite (mixed) states, and, moreover, it identifies a large class of states that are fully nonlocal. This is significant, as the only known example of a multipartite fully nonlocal mixed state is one composed of five copies of the Smolin state \cite{Smolin}, provided in Ref. \cite{Almeida_2010}.

\textit{Bounding the genuine nonlocality content.}
Lemma \eqref{lem:MFNL} relies on the proof of full nonlocality of the bipartite maximally entangled states which is achieved from violation of the chained Bell inequalities \cite{chained,Barrett_2006} in the limit of the number of local measurements taken to infinity. What is more, it does not account for noises and experimental imperfections. All this makes our results hardly testable in experiments. Here we address this issue by deriving a general bound on the genuine nonlocality content $\tilde{p}_{NL}$ 
of any behavior that tolerates noises and applies to scenarios with a finite number of local measurements.

\begin{thm}\label{thm:inequality}
Given a state $\rho$ and two parties $\alpha$ and $\bar{\alpha}$,
let us denote $\tilde{p}_{NL}^{\alpha|\overline{\alpha}}=\max_{\mathbf{x}_{R}} \min_{\mathbf{a}_{R}}\tilde{p}_{NL}^{\alpha|\overline{\alpha}}(\mathbf{a}_{R},\mathbf{x}_{R})$, where $\tilde{p}_{NL}^{\alpha|\overline{\alpha}}(\mathbf{a}_{R},\mathbf{x}_{R})$ is the minimal nonlocality content $p_{NL}$ of a bipartite state shared by $\alpha,\overline{\alpha}$ that was created by performing local measurements $\mathbf{x}_{R}$ by the remaining parties $R=[N]\setminus \{\alpha,\overline{\alpha}\}$ and corresponding to the outcomes $\mathbf{a}_{R}$. Then, $\tilde{p}_{NL}$ of $\mathcal{P}$ is lower-bounded as
\begin{equation}\label{ChanChan}
\tilde{p}_{NL}\geqslant 1 - \frac{1}{N-1}\sum_{\alpha=1}^{N}\sum_{\overline{\alpha}>\alpha}\left(1-\tilde{p}_{NL}^{\alpha|\overline{\alpha}}\right).
\end{equation}
\end{thm}
We present the proof in Appendix \ref{app:ineq} Note, that the above can also be used to lower bound $\tilde{p}_{NL}$ of a state or a subspace $V\subseteq \mathcal{H}$; in the latter case one needs to minimize $\tilde{p}_{NL}^{\alpha|\overline{\alpha}}$ over all $\rho\in\mathcal{B}(V)$.

There are multiple of ways in which Theorem \ref{thm:inequality} can be used. First, together with the results of Ref. \cite{Barrett_2006}, it enables determining the minimal number of measurements necessary to detect GMNL of $N$-partite systems by violating the chained Bell inequality between every pair of parties.
Let us consider for instance the two-dimensional subspace $V_5$ corresponding to the five-qubit code. Assuming that $\tilde{p}_{NL}^{\alpha|\overline{\alpha}}$ is the same for all pairs $\alpha>\overline{\alpha}$, it follows that to detect GMNL ($\tilde{p}_{NL}>0$) of this subspace one needs $p_{NL}^{\alpha|\overline{\alpha}}>3/5$. This can be achieved from violation of the chained Bell inequality with at least four measurements per observer (see Appendix \ref{app:ineq} for more details).

On the other hand, Theorem \ref{thm:inequality} allows for an estimation of genuine nonlocality content $\tilde{p}_{NL}$ from the experimentally observed bipartite $\tilde{p}_{NL}^{\alpha|\overline{\alpha}}$.
%
For instance, in Ref. \cite{PhysRevX.5.041052} the value of $\tilde{p}_{NL}^{\alpha|\overline{\alpha}}=0.874 \pm 0.001$ has been achieved experimentally for the two-qubit maximally entangled state. 
Assuming that such $\tilde{p}_{NL}^{\alpha|\overline{\alpha}}$ could be achieved for every pair of qubits, 
this would imply that the genuine nonlocality content of $V_5$ is
%
%
$\tilde{p}_{NL}\gtrsim 0.685$.

\textit{Qudit graph states.}
A natural extension of Theorem \ref{thm:mfnl} would be to consider the stabilizer formalism with all local systems being $d$-dimensional. Unfortunately, it turns out that one cannot directly generalize Lemma \ref{lem:two_qubits} to higher $d>2$. Nevertheless, using our approach we can still prove that simplest one-dimensional qudit stabilizer subspaces, which are in fact local-unitarily equivalent to multiqudit graph states, we can conclude that they are all MFNL; the same conclusion for graph states $d=2$ was derived before in  Ref. \cite{Almeida_2010}.

The qudit graph states are defined as follows: let $G$ be a multigraph. Then a graph state $\ket{G}$ associated to $G$ is one that is stabilized by $\mathbb{S}_{G}=\langle g_{1},\dots,g_{N}\rangle$ with
\begin{equation}\label{eq:graph}
g_{j}=\mathbf{X}_{j} \prod_{l=1}^{N} \mathbf{Z}_{l}^{\Gamma_{j,l}},
\end{equation}
where $\Gamma_{j,l}$ denotes the number of edges connecting vertices $j$ and $l$ in the graph $G$ and
\begin{equation}
\mathbf{X}=\sum_{j=0}^{d-1}\ket{j+1}\!\! \bra{j},\quad \mathbf{Z}= \sum_{j=0}^{d-1} \exp(2\pi \mathbb{i} j/d) \ket{j}\!\! \bra{j}
\end{equation}
are the generalized Pauli matrices. Notice that each qudit of $\ket{G}$ is associated with a vertex in $G$.

It is known that a graph state $\ket{G}$ is GME iff the graph $G$ is connected. This implies that for every bipartition $Q|\overline{Q}$, there exists a pair of vertices $i,j$ such that $i\in Q$, $j\in \overline{Q}$ and $\Gamma_{i,j}\neq 0$. To show that each GME graph state is MFNL, we begin by performing measurements in the computational basis $\{\ket{j}\}_{j=0}^{d-1}$ on every qudit apart from $i$ and $j$. One then finds that the post-measurement state $\ket{\psi_{i,j}}$ shared by parties $i$ and $j$ is stabilized by $\mathbf{X}_{i} \mathbf{Z}_{j}^{\Gamma_{i,j}}$ and $\mathbf{Z}_{i}^{\Gamma_{i,j}} \mathbf{X}_{j}$. One checks that these operators uniquely identify the state $\ket{\psi}$ as a maximally entangled state $\ket{\phi_{+}}$ \eqref{eq:max} (again up to local unitaries) for $q=d/r$, where $r$ is the greatest common divisor of $d$ and $\Gamma_{i,j}$. Since this procedure can be performed for any bipartition, Lemma \ref{lem:MFNL} allows us to conclude that every qudit graph state is MFNL.

\textit{Outlook and discussion.} There are still a few open questions to be explored. The most obvious one is whether qudit stabilizer subspaces that are genuinely multipartite entangled, are also multipartite fully nonlocal. As we discussed above, the approach used here which makes use of Lemma \ref{lem:two_qubits} is not suitable for that purpose, however, it is likely that this relationship between genuine multipartite entanglement and multipartite full nonlocality persists for any local dimension. 

Alternatively, if one were interested only in testing genuine multipartite nonlocality, then a possibly better strategy would be to directly construct suitable Bell inequalities detecting it, such as those provided in Refs. \cite{Bancal_2012,Aolita,Curchod_2019}.

Another question to explore is whether it is possible to experimentally determine the genuine nonlocality content of multipartite stabilizer mixed states, building on our results. Multiple experiments aiming to determine the nonlocality content of bipartite states have already been performed (see, e.g., Refs. \cite{PhysRevLett.95.240406,Gallego1,PhysRevX.5.041052, PhysRevLett.118.130403} and references therein). The question is thus whether they can be combined with our Theorem \ref{thm:inequality} to provide results for the multipartite scenario.

One can finally ask whether a tighter bound on the genuine nonlocality content than that in Eq. (\ref{ChanChan}) can be obtained, for instance throughout including genuine nonlocality contents of $m$-partite states $(m<N)$ created by the remaining parties.

\begin{acknowledgments}
\textit{Note added.}
While finishing this manuscript, we became aware that another group is independently working on a related question (see "Conclusion and outlook" in Ref. \cite{Zwerger_2019}).

\textit{Acknowledgments.}
We thank Ignacy Stachura and Błażej Kuzaka for helpful discussions. This work is supported by 
the National Science Centre (Poland) through the SONATA BIS project No. 2019/34/E/ST2/00369. This project has received funding from the European Union’s Horizon Europe research and innovation programme under grant agreement No 101080086 NeQST.
\end{acknowledgments}

\bibliographystyle{apsrev4-1}
%

\onecolumngrid
\appendix

\section{Proof of Lemma \ref{lem:two_qubits}}\label{app:two_qubits}

To start, let us introduce the formalism used in the proof of Lemma \ref{lem:two_qubits}, i.e., the formalism of commutation matrices \cite{Englbrecht_2022}. Let us consider a stabilizer $\mathbb{S}=\langle g_{1},\dots,g_{k}\rangle$ and some bipartition $Q|\overline{Q}$. Using the fact that $ZX=- XZ$ we can write
\begin{equation}
g_{i}^{(Q)}g_{j}^{(Q)}g_{i}^{(Q)}g_{j}^{(Q)}=(-1)^{c_{i,j}^{Q}}\mathbb{1},
\end{equation}
where $c_{i,j}^{Q}\in\{0,1\}$ depends on the exact form of $g_{i}$ and $g_{j}$. The commutation matrix $C^{Q}$ corresponding to a bipartition $Q$ is constructed in the following way
\begin{equation}
\left(C^{Q}\right)_{i,j}=c_{i,j}^{Q}.
\end{equation}
While this is a general construction for any set $Q$, for our purposes we only need $Q$ to be a singular party. Therefore, to make the notation easier, we will drop the "set" notation from the upper index, i.e., commutation matrix $C^{\alpha}$ corresponds to the party $\alpha$. With this notation, we formulate the following fact which corresponds to Eq. (26) in \cite{Englbrecht_2022}.
\begin{fact}\label{fact:com_matrices}
For the set of commutation matrices corresponding to a stabilizer with a local dimension $2$ we have
\begin{equation}
\sum_{\alpha\in\{1,\dots,N\}}C^{\alpha}=0,
\end{equation}
where the addition is modulo $2$.
\end{fact}

Given a stabilizer $\mathbb{S}$, there exist many equivalent choices generating sets that generate $\mathbb{S}$, for example  $\langle X\otimes X, Z\otimes Z\rangle=\langle X\otimes X, XZ\otimes XZ\rangle$. The transformations between one generating set and another are represented, in the commutation matrices formalism, by matrices $A\in M_{k\times k}(\mathbb{Z}_{d})$ acting on the commutation matrix $C^{\alpha}$
\begin{equation}\label{eq:commutation_transformation}
\tilde{C}^{\alpha}=A^{T}C^{\alpha} A,
\end{equation}
where $A$ has to be invertible for this to be a proper transformation between two generating sets of the same stabilizer, and the operations are performed modulo $2$.

As the proof relies on the commutation matrices, we need to first translate the GME condition in the language of this formalism.
\setcounter{lem}{4}
\begin{lem}\label{lem:gme_com}
Let $\mathbb{S}=\langle g_{1},\dots,g_{k}\rangle$ be a stabilizer with a stabilizer subspace $V_{\mathcal{S}}$, and let $C^{\alpha}$ be a set of commutation matrices associated to the generating set $\{g_{i}\}_{i=1}^{k}$. If $V_{\mathbb{S}}$ is GME then this implies that for all $Q\subset [N]$ we have
\begin{equation}
\sum_{\alpha\in Q} C^{\alpha} \neq 0.
\end{equation}
\end{lem}
\begin{proof}
We will show this by the contradiction. Let us assume that $V_{\mathbb{S}}$ is GME and there exists $Q\subset [N]$ such that
\begin{equation}
    \sum_{\alpha\in Q} C^{\alpha} =0.
\end{equation}
This implies that for all $i,j\in [k]$ we have
\begin{equation}
\left[g_{i}^{(Q)},g_{j}^{(Q)}\right]=0,
\end{equation}
which contradicts Lemma 1 from the main text.
\end{proof}

Let us recall here Lemma \ref{lem:two_qubits} after which we proceed with the proof.
\setcounter{lem}{2}
\begin{lem}
Let $\mathbb{S}$ be a stabilizer. Stabilizer subspace $V_{\mathbb{S}}$ is GME iff for every pair of qubits $\alpha_{1},\alpha_{2}\in [N]$ there exists a pair of stabilizing operators $s_{i},s_{j}\in\mathbb{S}$ such that
\begin{equation}
\left[s_{i}^{(\alpha_{l})},s_{j}^{(\alpha_{l})}\right]\neq0, \quad \left[s_{i}^{(\alpha)},s_{j}^{(\alpha)}\right]=0
\end{equation}
for all $[N]\setminus\{\alpha_{1},\alpha_{2}\}$ and all $l\in \{1,2\}$.
\end{lem}
\begin{proof}
The implication $\Leftarrow$ is triaival by virtue of Lemma 1, so here we only focus on the implication $\Rightarrow$. Let $\mathbb{S}=\langle g_{1},\dots, g_{k}\rangle$ be a stabilizer of a GME subspace $V_{\mathbb{S}}$, and $\{C^{\alpha}\}_{\alpha=1}^{N}$ be a commutation matrix set associated to the generating set $\{g_{i}\}_{i=1}^{k}$. The first step of the proof is to describe the condition in the theorem in terms of the commutation matrix formalism. To this end, we make use of the transformation of the commutation matrices described by Eq. \eqref{eq:commutation_transformation}.

Then the condition of the theorem can be reprised as follows: given the set of matrices $\{C^{\alpha}\}_{\alpha=1}^{N}$ corresponding to a GME stabilizer subspace, for all pairs of qubits $\alpha_{1},\alpha_{2}\in [N]$ there exists a transformation matrix $A$ and a pair of indices $i,j\in [k]$, such that
\begin{align}
\begin{split}
(A^{T} C^{\alpha_{1}} A)_{i,j} &= 1= (A^{T} C^{\alpha_{2}} A)_{i,j} \qquad (A^{T} C^{\alpha} A)_{i,j} =0, \quad \textrm{for all } \alpha \in [N]\setminus \{\alpha_{1},\alpha_{2}\}.
\end{split}
\end{align}

Without a loss of generality, we can take $\alpha_{1}=1$, $\alpha_{2}=N$. From Fact \ref{fact:com_matrices} it follows that $C^{N}$ can be expressed as a sum of the rest of the matrices $C^{\alpha}$, and so we can disregard it, leading us to the condition
\begin{equation}\label{eq:condition_mfnl_proof}
(A^{T} C^{\alpha} A)_{i,j} =\begin{cases}
1 &\textrm{for } \alpha =1,\\
0 &\textrm{for } \alpha\in \{2,\dots,N-1\}.
\end{cases}
\end{equation}
Let $a_{ij}=(A)_{i,j}$ and $c^{\alpha}_{i,j}=(C^{\alpha})_{i,j}$. Then the condition \eqref{eq:condition_mfnl_proof} can be written as
\begin{equation}
\sum_{m,l=1}^{k}a_{m,i}c_{m,l}^{\alpha}a_{l,j}= \delta_{\alpha,1}
\end{equation}
for all $\alpha\in [N]$, where $\delta_{\alpha,1}$ is a Kronecker delta and operations are performed modulo $2$. Notice, that this condition describes an action of two vectors on matrices $C^{\alpha}$. Let
\begin{equation}
v=(a_{1,j},a_{2,j},\dots,a_{k,j})^{T},\quad u=(a_{1,i},a_{2,i},\dots,a_{k,i})^{T}
\end{equation}
This allows us to rewrite the condition as: for all $C^{\alpha}$ there exists a pair of vectors $u,v\in \mathbb{Z}_{2}^{k}$ such that.
\begin{equation}\label{eq:ucv}
u^{T}C^{\alpha}v=\delta_{\alpha,1}.
\end{equation}
For $A$ to be a proper basis change it has to be invertible, which necessitates that $v\neq u$ and $v\neq 0 \neq u$. The latter is fulfilled because $v=0$ imples $u^{T}C^{\alpha}v=0$ for all $u$ and $C^{\alpha}$ (similarly for $u=0$).  As for the $v\neq u$, if we set $v=u$ Eq. \eqref{eq:ucv} equals
\begin{equation}
\sum_{i,j=1}^{k}u_{i}C_{i,j}^{\alpha}u_{j}=\sum_{i=1}^{k}\sum_{j=i}^{k} (u_{i}C_{i,j}^{\alpha}u_{j}-u_{j}C_{i,j}^{\alpha}u_{i})=0,
\end{equation}
where we used the property $(C^{\alpha})_{j,i}=(C^{\alpha})_{i,j}$ and again operations are performed modulo $2$. This again implies that for all $u$ we have $u^{T} C^{\alpha}u=0$. Therefore, the condition \eqref{eq:ucv} implicitly implies that $A$ is invertible.

We can now proceed with the proof proper, i.e., we want to show that from the assumption about $V_{\mathbb{S}}$ being GME implies that \eqref{eq:ucv} holds true for some $u$ and $v$. To this end, we first have to show that for vectors $v^{\alpha}=C^{\alpha}v$, it follows form Lemam \ref{lem:gme_com} that for all $\beta\in [N-1]$ there exists a vector $v$ such that 
\begin{equation}\label{eq:v^beta}
v^{\beta}\neq \sum_{\alpha\in I_{\beta}}  v^{\alpha}
\end{equation}
for all $I_{\beta}\subset [N-1]\setminus \{\beta\}$.

Let us assume the converse, i.e., that there exists $\beta$ such that for all $v$ there exists a choice of $I_{\beta}$ such that
\begin{equation}
v^{\beta}=\sum_{\alpha\in I_{\beta}} v^{\alpha}.
\end{equation}
From the definition of $v^{\alpha}$ we have that
\begin{equation}\label{eq:c_v}
C^{\beta} v = \sum_{\alpha\in I_{\beta}}C^{\alpha}v.
\end{equation}
The fact that this would have to be true for all $v$ implies that
\begin{equation}\label{eq:c}
C^{\beta} = \sum_{\alpha\in I_{\beta}}C^{\alpha},
\end{equation}
which contradicts Lemma \ref{lem:gme_com} together with the assumption that $V_{\mathbb{S}}$ is GME.

With Eq. \eqref{eq:v^beta} proven, let us choose $v$ such that 
\begin{equation}\label{eq:v^1}
v^{1}\neq \sum_{\alpha\in I_{1}}v^{\alpha}
\end{equation}
for all $I_{1}\subset\{2,\dots,N-1\}$. The above construction requires the choice of a specific $v$, but keeping in mind that our goal is to prove Eq. \eqref{eq:ucv}, we still have freedom of choice in regards to $u$. 

Notice, that vectors $\{v^{\alpha}\}_{\alpha=1}^{N-1}$ span a subspace. Let $\{w_{j}\}_{j=1}^{q}$ be a basis chosen from $\{v^{\alpha}\}_{\alpha=1}^{N-1}$ spanning this subspace of dimension $q$, s.t. $w_{1}=v^{1}$. From Eq. \eqref{eq:v^1} and the choice of basis $\{w_{j}\}_{j=1}^{q}$ it follows that for all $J\subset \{2,\dots,q\}$
\begin{equation}
v^{1}\neq \sum_{j\in J}w_{j}
\end{equation}
and also that for all $\alpha \in \{2,\dots,N-1\}$ there exists a subset $J\subset \{2,\dots,q\}$ such that
\begin{equation} \label{eq:w}
v^{\alpha} = \sum_{j\in J} w_{j}.
\end{equation}
Let us define a transformation matrix $T$ by the following relations:
\begin{equation}
Tw_{j}=e_{j}\quad \textrm{for all }j\in [q],
\end{equation}
where $e_{j}$ is a unit vector with one on the $j$'th entry and $0$ elsewhere. Let us take $\tilde{u}=e_{1}$. Then
\begin{equation}
\tilde{u}^{T} T w_{1} = 1,\qquad   \tilde{u}^{T} T w_{j} = 0 \quad\textrm{for all }j\in\{2,3,\dots,q\}.
\end{equation}
Therefore taking $u=T^{T}\tilde{u}=T^{T}e_{1}$ gives us the desired relations
\begin{equation}
u^{T} w_{j}= \delta_{1,j}.
\end{equation}
From the above, it follows
\begin{equation}
u^{T}\sum_{j\in J} w_{j}=0,
\end{equation}
where $J\subset \{2,\dots,q\}$. This together with Eq. \eqref{eq:w} implies that for all $\alpha \in [N]$
\begin{equation}
u^{T}C^{\alpha}v=u^{T}v^{\alpha}=e_{2}T^{T}v^{\alpha}=\delta_{\alpha,1},
\end{equation}
which ends the proof.
\end{proof}

\section{Proof of Theorem \ref{thm:mfnl}}\label{app:gme_gmnl}
\setcounter{thm}{0}
\begin{thm}
A stabilizer subspace $V_{\mathbb{S}}$ is multipartite fully nonlocal iff it is genuinely multipartite entangled.
\end{thm}
\begin{proof}
Trivially, if a subspace is MFNL it implies that it is GMNL, and so it is also GME. Therefore, let us focus on the proof of the implication $\Leftarrow$. Let us consider a stabilizer $\mathbb{S}$ fulfilling the conditions of Lemma \ref{lem:two_qubits}. By the virtue of Eq. \eqref{eq:lem_ncom} for a given bipartition $Q|\overline{Q}$, $s_{i}^{(\alpha)}$ and $s_{j}^{(\alpha)}$ have a common eigenbasis for all $\alpha\in [N]\setminus\{\alpha_{1},\alpha_{2}\}$. Let us denote by $\{\Pi_{a_{\alpha}}^{\alpha}\}_{a_{\alpha}=0}^{1}$ a set of projective measurement operators onto the common eigenbasis of $s_{i}^{(\alpha)}$ and $s_{j}^{(\alpha)}$, where $a_{\alpha}$ is the measurement result.

Let us consider a state $\rho\in\mathcal{B}(V_{\mathbb{S}})$. After performing local measurements $\Pi_{a_{\alpha}}^{\alpha}$ the resulting state $\sigma_{a_{1},\dots,a_{N}}$ shared by parties $\alpha_{1}$ and $\alpha_{2}$ equals
\begin{equation}
\sigma_{a_{1},\dots,a_{N}}=\frac{1}{n}\operatorname{Tr}_{A}\left[\bigotimes_{\alpha=1}^{N}\Pi_{a_{\alpha}}^{\alpha} \cdot \rho\right],
\end{equation}
where we take $\Pi_{a_{\alpha_{1}}}^{\alpha_{1}}=\Pi_{a_{\alpha_{2}}}^{\alpha_{2}}=\mathbb{1}$, $n$ is a normalization constant, and $\operatorname{Tr}_{A}$ is a partial trace over a set of parties $A=[N]\setminus\{\alpha_{1},\alpha_{2}\}$. From $\rho\in\mathcal{B}(V_{\mathbb{S}})$ it follows that
\begin{align}
\begin{split}
\sigma_{a_{1},\dots,a_{N}}
&=\frac{1}{n}\operatorname{Tr}_{A}\left[\bigotimes_{\alpha=1}^{N}\Pi_{a_{\alpha}}^{\alpha} \cdot s_{i}\rho\right]
=\left(\prod_{\alpha \in A}(-1)^{\tau_{i}(a_{\alpha})}\right) s_{i}^{(\{\alpha_{1},\alpha_{2}\})} \frac{1}{n}\operatorname{Tr}_{A}\left[\bigotimes_{\alpha=1}^{N}\Pi_{a_{\alpha}}^{\alpha} \cdot \rho\right]\\
&=\left(\prod_{\alpha \in A}(-1)^{\tau_{i}(a_{\alpha})}\right) s_{i}^{(\{\alpha_{1},\alpha_{2}\})}\sigma_{a_{1},\dots,a_{N}}
\end{split}
\end{align}
where the second equality uses the fact that from the measurement $\Pi_{a_{\alpha}}^{\alpha}$ are performed in the eigenbasis of $s_{i}^{(\alpha)}$, which implies
\begin{equation}
 s_{i}^{(\alpha)}\Pi_{a_{\alpha}}^{\alpha}= (-1)^{\tau_{i}(a_{\alpha})}\Pi_{a_{\alpha}}^{\alpha}
\end{equation}
for $\tau_{i}(a_{\alpha})\in \{0,1\}$ which is a function of both $s_{i}^{\alpha}$ and the outcome $a_{\alpha}$. The same can be done for the operator $s_{j}$ which gives us two stabilizing relations for $\sigma_{a_{1},\dots,a_{N}}$
\begin{align}
\begin{split}
\sigma_{a_{1},\dots,a_{N}} &= \left(\prod_{\alpha \in A}(-1)^{\tau_{i}(a_{\alpha})}\right) s_{i}^{(\{\alpha_{1},\alpha_{2}\})} \sigma_{a_{1},\dots,a_{N}}=\tilde{s}_{i}\sigma_{a_{1},\dots,a_{N}},\\
\sigma_{a_{1},\dots,a_{N}} &=\left(\prod_{\alpha \in A}(-1)^{\tau_{j}(a_{\alpha})}\right) s_{j}^{(\{\alpha_{1},\alpha_{2}\})} \sigma_{a_{1},\dots,a_{N}}=\tilde{s}_{j}\sigma_{a_{1},\dots,a_{N}},
\end{split}
\end{align}
where we omit the dependence of $\tilde{s}_{i}$, $\tilde{s}_{j}$ on the measurement results $a_{1},\dots,a_{N}$ to simplify the notation. By the virtue of Eq. \eqref{eq:lem_ncom} we have that $\tilde{s}_{i}$ and $\tilde{s}_{j}$ have to be independent. Then a stabilizer subspace $V_{\mathbb{S}}$ corresponding to the stabilizer $\mathbb{S}=\langle \tilde{s}_{i},\tilde{s}_{j}\rangle$ is one-dimensional \cite{GHEORGHIU2014505}, i.e., the above relations uniquely identify the state $\sigma_{a_{1},\dots,a_{N}}$.

Since $\tilde{s}_{i}$ and $\tilde{s}_{j}$ are a two-fold tensor product of Pauli matrices, such that $\tilde{s}_{i}^{\alpha_{1}}$ and $\tilde{s}_{j}^{\alpha_{1}}$ anticommute, we can easily find unitaries $U_{\alpha_{1}}$ and $U_{\alpha_{2}}$ such that
\begin{equation}
U_{\alpha_{1}}\otimes U_{\alpha_{2}} \tilde{s}_{i} U_{\alpha_{1}}^{\dagger}\otimes U_{\alpha_{2}}^{\dagger}=X\otimes X,\qquad U_{\alpha_{1}}\otimes U_{\alpha_{2}} \tilde{s}_{j} U_{\alpha_{1}}^{\dagger}\otimes U_{\alpha_{2}}^{\dagger}=Z\otimes Z,
\end{equation}
therefore
\begin{equation}
U_{\alpha_{1}}\otimes U_{\alpha_{2}} \sigma_{a_{1},\dots,a_{N}} U_{\alpha_{1}}^{\dagger}\otimes U_{\alpha_{2}}^{\dagger}= \ket{\phi_{+}}\! \bra{\phi_{+}}.
\end{equation}
The last note is that since $\tilde{s}_{i}$ and $\tilde{s}_{j}$  depend on the measurement results $a_{1},\dots,a_{N}$, so do $U_{\alpha_{1}}$ and $U_{\alpha_{2}}$. Nonetheless, for all measurement results we get a maximally entangled state for all $\rho\in \mathcal{B}(V_{\mathbb{S}})$, and since this holds true for all pairs of parties $\alpha_{1},\alpha_{2}\in[N]$, we have, by virtue of Lemma \ref{lem:two_qubits}, that $V_{\mathcal{S}}$ is MFNL.
\end{proof}

Let us note here that Thorem \ref{thm:mfnl} cannot be generalized to stabilizers of a higher local dimension. A simple example of a stabilizer illustrating this problem is $\mathbb{S}=\langle \mathbf{X}\otimes \mathbf{X}\otimes \mathbf{X}, \mathbf{Z}\otimes \mathbf{Z}\otimes \mathbf{Z} \rangle$ for $d=3$, where
\begin{equation}
\mathbf{X}=\sum_{j=0}^{d-1}\ket{j+1}\!\! \bra{j},\quad \mathbf{Z}= \sum_{j=0}^{d-1} \exp(2\pi \mathbb{i} j/d) \ket{j}\!\! \bra{j}
\end{equation}
are the generalized Pauli matrices. By the result of Ref. \cite{Makuta2023fullynonpositive} [see Corollary 1 therein], $V_{\mathbb{S}}$ associated with this stabilizer is GME. On the other hand, it follows that for any two operators $s_{i},s_{j} \in \mathbb{S}$ we either have $[s_{i}^{(\alpha)},s_{j}^{(\alpha)}]=0$ or $[s_{i}^{(\alpha)},s_{j}^{(\alpha)}]\neq 0$ for all $\alpha\in\{1,2,3\}$, which disagrees with Lemma \ref{lem:two_qubits}. This of course does not imply that GME qudit stabilizer subspaces are not MFNL or GMNL, but rather it suggests that a different approach must be used for $d>2$. 

\section{Proof of Theorem \ref{thm:inequality}}\label{app:ineq}

In this section, we prove Theorem \ref{thm:inequality}. To this end, let us first recall the generalized EPR-2 decomposition \cite{EPR2,Almeida_2010}. Precisely, every behavior $\mathcal{P}=\{P(\mathbf{a}| \mathbf{x})\}_{\mathbf{a},\mathbf{x}}$ can be decomposed as
\begin{equation}\label{eq:behavior_decomposition}
P(\mathbf{a}| \mathbf{x})=\sum_{t\in T}p_{t} P_{t}(\mathbf{a}| \mathbf{x})+p_{NL}P_{NL}(\mathbf{a}| \mathbf{x}),
\end{equation}
where $T$ is the set of all possible partitions $t=\{S_{i}\}_{i=1}^{m}$ of the set $[N]$ into $m$ nonempty subsets $S_{i}$ for any $m\in\{2,\ldots,N\}$, $P_{NL}(\mathbf{a}|\mathbf{x})$ is a probability distribution that is nonlocal across any bipartition, and $\mathcal{P}_{t}(\mathbf{a}|\mathbf{x})$ is a probability distribution that is local across any bipartition $Q|\overline{Q}$ such that $Q=\bigcup_{i\in I} S_{i}$ for any subset $I\subset[m]$, and nonlocal across any other bipartition. We denote by $\tilde{p}_{NL}$ the minimal $p_{NL}$ over all such decompositions, and we use the notation $\tilde{p}_{t}$ for a decomposition such that $p_{NL}=\tilde{p}_{NL}$.

Next, let us consider two parties labelled $\alpha$ and $\bar{\alpha}$ and the corresponding behavior $\{P(a_{\alpha},a_{\overline{\alpha}}|\mathbf{a}_{R},\mathbf{x})\}_{a_{\alpha},a_{\overline{\alpha}},x_{\alpha},x_{\overline{\alpha}}}$ conditioned on the outcomes observed by the remaining parties $R=[N]\setminus\{\alpha,\bar{\alpha}\}$ upon performing measurements $\mathbf{x}_R$. We then consider an EPR-2 decomposition of that behavior obtained from Eq. \eqref{eq:behavior_decomposition}:
\begin{equation}\label{eq:epr-2_alpha}
P(a_{\alpha},a_{\overline{\alpha}}|\mathbf{a}_{R},\mathbf{x}) = p_{L}^{\alpha|\overline{\alpha}} P_{L}^{\alpha|\overline{\alpha}}(a_{\alpha},a_{\overline{\alpha}}|\mathbf{a}_{R},\mathbf{x}) + p_{NL}^{\alpha|\overline{\alpha}} P_{NL}^{\alpha|\overline{\alpha}}(a_{\alpha},a_{\overline{\alpha}}|\mathbf{a}_{R},\mathbf{x}),
\end{equation}
where $R=[N]\setminus\{\alpha,\overline{\alpha}\}$, $P_{L}^{\alpha|\overline{\alpha}}(\mathbf{a}| \mathbf{x})$ is a probability distribution local across any bipartition $Q|\overline{Q}$ such that $\alpha\in Q$, $\overline{\alpha}\in \overline{Q}$, $P_{NL}^{\alpha|\overline{\alpha}}(\mathbf{a}| \mathbf{x})$ is a probability distribution nonlocal across any bipartition $Q|\overline{Q}$ such that $\alpha\in Q$, $\overline{\alpha}\in \overline{Q}$, and $p_{L}^{\alpha|\overline{\alpha}}+p_{NL}^{\alpha|\overline{\alpha}}=1$.

We can now proceed with the proof of Theorem \ref{thm:inequality}, which we state below for convenience.
\begin{thm}
Consider a state $\rho$ and a pair of parties $\alpha,\bar{\alpha}$. Let then  
\begin{equation}
\tilde{p}_{NL}^{\alpha|\overline{\alpha}}=\max_{\mathbf{x}_{R}} \min_{\mathbf{a}_{R}}\tilde{p}_{NL}^{\alpha|\overline{\alpha}}(\mathbf{a}_{R},\mathbf{x}_{R}),
\end{equation}
where $\tilde{p}_{NL}^{\alpha|\overline{\alpha}}(\mathbf{a}_{R},\mathbf{x}_{R})$ is the nonlocality content $\tilde{p}_{NL}$ of a bipartite state shared by parties $\alpha,\overline{\alpha}$ 
that corresponds to the outcomes $\mathbf{a}_{R}$ observed by the remaining parties 
$R=[N]\setminus{\alpha,\bar{\alpha}}$ after performing the local measurements 
$\mathbf{x}_{R}$ on their shares of the state $\rho$. 
Then, $\tilde{p}_{NL}$ of $\rho$ is lower-bounded by
\begin{equation}
\tilde{p}_{NL}\geqslant 1 - \frac{1}{N-1}\sum_{\alpha=1}^{N}\sum_{\overline{\alpha}>\alpha}(1-\tilde{p}_{NL}^{\alpha|\overline{\alpha}}).
\end{equation}
\end{thm}
\begin{proof}
Let us choose two parties $\alpha,\overline{\alpha}\in [N]$ such that $\alpha \neq \overline{\alpha}$. Let us consider the decomposition \eqref{eq:behavior_decomposition} for which $p_{NL}=\tilde{p}_{NL}$. Then, we can sum Eq. \eqref{eq:behavior_decomposition} over all $a_{\alpha}, a_{\overline{\alpha}}$ to get
\begin{equation}\label{eq:behavior_decomposition_without_alpha}
P(\mathbf{a}_{R}| \mathbf{x}_{R})=\sum_{t\in T}\tilde{p}_{t} P_{t}(\mathbf{a}_{R}| \mathbf{x}_{R})+\tilde{p}_{NL}P_{NL}(\mathbf{a}_{R}| \mathbf{x}_{R}),
\end{equation}
where $R=[N]\setminus\{\alpha,\overline{\alpha}\}$, and we used the non-signaling principle to remove the dependence on $x_{\alpha},x_{\overline{\alpha}}$. On the other hand, for every $P(\mathbf{a}|\mathbf{x})$, including $P_{t}(\mathbf{a}|\mathbf{x})$ and $P_{NL}(\mathbf{a}|\mathbf{x})$, we have
\begin{equation}
P(\mathbf{a}| \mathbf{x})= P(a_{\alpha},a_{\overline{\alpha}}|\mathbf{a}_{R},\mathbf{x})P(\mathbf{a}_{R}| \mathbf{x}_{R}),
\end{equation}
where we once again made use of the non-signaling assumption. Substitution of the above to Eq. \eqref{eq:behavior_decomposition} yields
\begin{equation}\label{eq:behavior_decomposition_alpha}
P(a_{\alpha},a_{\overline{\alpha}}|\mathbf{a}_{R},\mathbf{x})=\frac{1}{P(\mathbf{a}_{R}|\mathbf{x}_{R})}\left(\sum_{t\in T}\tilde{p}_{t} P_{t}(a_{\alpha},a_{\overline{\alpha}}|\mathbf{a}_{R},\mathbf{x})P_{t}(\mathbf{a}_{R}|\mathbf{x}_{R})+\tilde{p}_{NL}P_{NL}(a_{\alpha},a_{\overline{\alpha}}|\mathbf{a}_{R},\mathbf{x})P_{NL}(\mathbf{a}_{R}|\mathbf{x}_{R})\right),
\end{equation}
which holds true for all $P(\mathbf{a}_{R}|\mathbf{x}_{R})\neq 0$. Next, let us take 
\begin{equation}\label{eq:prob_assignment}
P_{L}^{\alpha|\overline{\alpha}}(\mathbf{a}| \mathbf{x})= \frac{1}{\eta_{L}^{\alpha|\overline{\alpha}}}\sum_{t\in T_{L}^{\alpha,\overline{\alpha}}}\tilde{p}_{t} P_{t}(\mathbf{a}|\mathbf{x}),\qquad P_{NL}^{\alpha|\overline{\alpha}}(\mathbf{a}| \mathbf{x})= \frac{1}{\eta_{NL}^{\alpha|\overline{\alpha}}}\left(\sum_{t\in T_{NL}^{\alpha,\overline{\alpha}}}\tilde{p}_{t} P_{t}(\mathbf{a}|\mathbf{x}) + \tilde{p}_{NL}P_{NL}(\mathbf{a}|\mathbf{x})\right),
\end{equation}
where 
\begin{equation}
\eta_{L}^{\alpha|\overline{\alpha}}=\sum_{t\in T_{L}^{\alpha,\overline{\alpha}}}\tilde{p}_{t}, \qquad \eta_{NL}^{\alpha|\overline{\alpha}}=\sum_{t\in T_{NL}^{\alpha,\overline{\alpha}}}\tilde{p}_{t}+\tilde{p}_{NL},
\end{equation}
$T_{L}^{\alpha,\overline{\alpha}}$ is the set of partitions $t=\{S_{1},\ldots,S_{m}\}$ of $[N]$ such that $\alpha$ and $\overline{\alpha}$ are in different subsets $S_{i}$, and $T_{NL}^{\alpha,\overline{\alpha}}=T\setminus T_{L}^{\alpha,\overline{\alpha}}$. Notice, that $P_{L}^{\alpha|\overline{\alpha}}(\mathbf{a}| \mathbf{x})$ and $P_{NL}^{\alpha|\overline{\alpha}}(\mathbf{a}| \mathbf{x})$ are defined in a way as to satisfy the conditions of the decomposition \eqref{eq:epr-2_alpha}. Moreover, the assignment \eqref{eq:prob_assignment} allows us to rewrite \eqref{eq:behavior_decomposition} as
\begin{equation}\label{eq:induced_decomposition}
P(a_{\alpha},a_{\overline{\alpha}}|\mathbf{a}_{R},\mathbf{x})=\eta_{L}^{\alpha|\overline{\alpha}}\frac{P_{L}(\mathbf{a}_{R}|\mathbf{x}_{R})}{P(\mathbf{a}_{R}|\mathbf{x}_{R})} P_{L}^{\alpha|\overline{\alpha}}(a_{\alpha},a_{\overline{\alpha}}|\mathbf{a}_{R},\mathbf{x})+\eta_{NL}^{\alpha|\overline{\alpha}}\frac{P_{NL}(\mathbf{a}_{R}|\mathbf{x}_{R})}{P(\mathbf{a}_{R}|\mathbf{x}_{R})} P_{NL}^{\alpha|\overline{\alpha}}(a_{\alpha},a_{\overline{\alpha}}|\mathbf{a}_{R},\mathbf{x}),
\end{equation}
and since by virtue of Eq. \eqref{eq:behavior_decomposition_without_alpha} we have
\begin{equation}
\eta_{L}^{\alpha|\overline{\alpha}}\frac{P_{L}(\mathbf{a}_{R}|\mathbf{x}_{R})}{P(\mathbf{a}_{R}|\mathbf{x}_{R})}+\eta_{NL}^{\alpha|\overline{\alpha}}\frac{P_{NL}(\mathbf{a}_{R}|\mathbf{x}_{R})}{P(\mathbf{a}_{R}|\mathbf{x}_{R})}=1,
\end{equation}
Eq. \eqref{eq:induced_decomposition} is a proper decomposition \eqref{eq:epr-2_alpha} for a given $\mathbf{a}_{R}$ and $\mathbf{x}_{R}$. We then have
\begin{equation}
\tilde{p}_{L}^{\alpha|\overline{\alpha}}(\mathbf{a}_{R},\mathbf{x}_{R}) \geqslant\eta_{L}^{\alpha|\overline{\alpha}}\frac{P_{L}(\mathbf{a}_{R}|\mathbf{x}_{R})}{P(\mathbf{a}_{R}|\mathbf{x}_{R})},
\end{equation}
where by $\tilde{p}_{L}^{\alpha|\overline{\alpha}}(\mathbf{a}_{R},\mathbf{x}_{R})$ we specify that the maximal $p_{L}^{\alpha|\overline{\alpha}}$ depends on $\mathbf{a}_{R}$ and $\mathbf{x}_{R}$. Trivially, we have
\begin{equation}
\tilde{p}_{L}^{\alpha|\overline{\alpha}}(\mathbf{x}_{R})\coloneqq \max_{\mathbf{a}_{R}} \tilde{p}_{L}^{\alpha|\overline{\alpha}}(\mathbf{a}_{R},\mathbf{x}_{R})\geqslant \eta_{L}^{\alpha|\overline{\alpha}}\frac{P_{L}(\mathbf{a}_{R}|\mathbf{x}_{R})}{P(\mathbf{a}_{R}|\mathbf{x}_{R})}.
\end{equation}
We multiply the above inequality by $P(\mathbf{a}_{R}|\mathbf{x}_{R})$ and sum over $\mathbf{a}_{R}$ which yields
\begin{equation}
\tilde{p}_{L}^{\alpha|\overline{\alpha}}(\mathbf{x}_{R}) \geqslant\eta_{L}^{\alpha|\overline{\alpha}}.
\end{equation}
This inequality holds for every $\mathbf{x}_{R}$, and since the right side does not depend on $\mathbf{x}_{R}$, we can minimize over $\mathbf{x}_{R}$ to get
\begin{equation}\label{ineq:p_L}
\tilde{p}_{L}^{\alpha|\overline{\alpha}}\coloneqq \min_{\mathbf{x}_{R}}\max_{\mathbf{a}_{R}} \tilde{p}_{L}^{\alpha|\overline{\alpha}}(\mathbf{a}_{R},\mathbf{x}_{R})\geqslant \eta_{L}^{\alpha|\overline{\alpha}}=\sum_{t\in T_{L}^{\alpha,\overline{\alpha}}}\tilde{p}_{t}.
\end{equation}
Let us take $\overline{\alpha}=1$ and let us sum the above inequality over all $\alpha\in\{2,\ldots,N\}$
\begin{equation}
\sum_{\alpha=2}^{N}\tilde{p}_{L}^{\alpha|1}\geqslant \sum_{\alpha=2}^{N}\sum_{t\in T_{L}^{\alpha,1}}\tilde{p}_{t}=(N-1)(1-\tilde{p}_{NL}) - \sum_{\alpha=2}^{N}\sum_{t\in T_{NL}^{\alpha,1}}\tilde{p}_{t},
\end{equation}
where the equality follows from the normalization condition. We want to bound the second term on the right-hand side of the above inequality by the function of coefficients $\tilde{p}_{L}^{\alpha|\overline{\alpha}}$. To this end, we first need to show that
\begin{equation}\label{ineq:sums}
\sum_{\alpha=2}^{N}\sum_{t\in T_{NL}^{\alpha,1}}\tilde{p}_{t} \leqslant \sum_{\alpha=2}^{N}\sum_{\overline{\alpha}>\alpha}\sum_{t\in T_{L}^{\alpha,\overline{\alpha}}}\tilde{p}_{t}.
\end{equation}
Let us focus on the individual terms $\tilde{p}_{t}$.  For a given $t=\{S_{1},\ldots,S_{m}\}$ without the loss of generality we can take $1\in S_{1}$. Then $\tilde{p}_{t}$ appears on the left-hand side of the inequality $(|S_{1}|-1)$ times - once for each $\alpha\in S_{1}\setminus \{ 1\}$. On the other hand, the same term $\tilde{p}_{t}$ appears on the right-hand side of the inequality $(|S_{1}|-1)(N-|S_{1}|)$ times - once for each pair $\alpha,\overline{\alpha}$ such that $\alpha \in S_{1}\setminus \{1\}$, and $\overline{\alpha}\notin  S_{1}$. Since $1\leqslant |S_{1}|\leqslant N-1$, it follows that $(|S_{1}|-1)\leqslant(|S_{1}|-1)(N-|S_{1}|)$ for each $\tilde{p}_{t}$ which implies that Ineq. \eqref{ineq:sums} holds true. Therefore, we have
\begin{equation}
\sum_{\alpha=2}^{N}\tilde{p}_{L}^{\alpha|1}\geqslant(N-1)(1-\tilde{p}_{NL}) - \sum_{\alpha=2}^{N}\sum_{\overline{\alpha}>\alpha}\sum_{t\in T_{L}^{\alpha,\overline{\alpha}}}\tilde{p}_{t}
\geqslant (N-1)(1-\tilde{p}_{NL})-\sum_{\alpha=2}^{N}\sum_{\overline{\alpha}>\alpha}\tilde{p}_{L}^{\alpha|\overline{\alpha}},
\end{equation}
where the second inequality follows from Ineq. \eqref{ineq:p_L}. Substituting $\tilde{p}_{L}^{\alpha|\overline{\alpha}}=1-\tilde{p}_{NL}^{\alpha|\overline{\alpha}}$ we arrive at the lower bound of $\tilde{p}_{NL}$ expressed as a function of $\tilde{p}_{NL}^{\alpha|\overline{\alpha}}$.
\begin{eqnarray}
\tilde{p}_{NL}&\geqslant& 1 - \frac{1}{N-1}\sum_{\alpha=1}^{N}\sum_{\overline{\alpha}>\alpha}(1-\tilde{p}_{NL}^{\alpha|\overline{\alpha}})\nonumber\\
&=&\frac{1}{N-1}\sum_{\alpha=1}^{N}\sum_{\overline{\alpha}>\alpha}\tilde{p}_{NL}^{\alpha|\overline{\alpha}}-\left(\frac{N}{2}-1\right).
\end{eqnarray}
\end{proof}

Let us now discuss the application of Theorem \ref{thm:inequality}. In order to design an experiment utilizing the bound provided by this theorem it would be convenient to have a simple way of calculating $p_{NL}^{\alpha|\overline{\alpha}}$ directly from the experimental data. What is more, we also need to find a choice of measurement settings $\mathbf{x}$ that would give us a high value of $p_{NL}^{\alpha|\overline{\alpha}}$. Both of these conditions are met by the chained Bell inequality \cite{chained} 
\begin{equation}
I_{n,d} = \sum_{j=1}^{n-1}(\langle [A_{j}-B_{j}]\rangle +\langle [B_{j}-A_{j+1}]\rangle) +\langle [A_{n}-B_{n}]\rangle + \langle [B_{n}-A_{1}-1]\rangle \geqslant d-1,
\end{equation}
where $\langle M \rangle = \sum_{i=1}^{d-1}iP(M=i)$ and $[M]$ denotes $M \mod d$. In Ref. \cite{Barrett_2006} the authors found the optimal measurement strategy for violation of this inequality by the maximally entangled state $\ket{\phi_{+}}=1/\sqrt{d}\sum_{j=0}^{d-1}\ket{jj}$, and they discovered that the value of the Bell functional $I_{n,d}$ constitutes a lower bound on the factor $\tilde{p}_{NL}^{\alpha|\overline{\alpha}}$
\begin{equation}
\tilde{p}_{NL}^{\alpha|\overline{\alpha}}\geqslant1- \frac{I_{n,d}}{d-1}.
\end{equation}
We can use this inequality to find a lower bound on the multipartite $\tilde{p}_{NL}$ given the value of $I_{n,d}^{\alpha|\overline{\alpha}}$ originating form a chained Bell test each pair of qudits $\alpha|\overline{\alpha}$
\begin{equation}\label{ineq:p_NL>chained}
\tilde{p}_{NL}\geqslant 1 - \frac{1}{(N-1)(d-1)}\sum_{\alpha=1}^{N}\sum_{\overline{\alpha}>\alpha}I_{n,d}^{\alpha|\overline{\alpha}}.
\end{equation}

One interesting analysis that can be performed using this bound is to find a correspondence between the total number of measurement settings $m\geqslant n$ per party required for detection of GMNL, i.e., $\tilde{p}_{NL}>0$. To this end, we compute the maximal value of $I_{n,d}$ for each $n$, and assuming that this maximum is the same for each pair $\alpha|\overline{\alpha}$, we look for minimal $n$ such that
\begin{equation}
I_{n,d}<\frac{2(d-1)}{N},
\end{equation}
which is the value of $I_{n,d}$ that implies $\tilde{p}_{NL}>0$. 

Let us consider the case of measurement performed on states from a qubit stabilizer subspaces. For this specific example, we can determine the exact relationship between $m$ and $n$. Importantly for this derivation, the measurements that lead to the maximal violation of chained Bell inequality by the maximally entangled state are not symmetric under the exchange of parties. Therefore, to perform the chained Bell test for each pair $\alpha|\overline{\alpha}$ each party has to have access to both sets of measurements (apart from one party which can have only one of the sets of measurements). Additionally, each party has to have access to the Pauli measurements that are required for the generation of the maximally entangled state. Hence, for a chained Bell inequality for a given $n$ the total number of measurements required for detection of GMNL equals $m=2n+3$. See Fig. \ref{fig} where we computed the minimal $m$ that gives $\tilde{p}_{NL}$ for qubit stabilizer subspaces, as a function of $N\in\{4,\ldots,40\}$. 

Note, that we do not have to minimize the value of $\tilde{p}_{NL}$ over the states from the studied subspace as, under our chosen strategy, $\tilde{p}_{NL}$ will be the same for every state from the qubit stabilizer subspace. This of course follows from the fact that we can generate a maximally entangled state between any two qubits using any state from the subspace, and so $\tilde{p}_{NL}^{\alpha|\overline{\alpha}}$ is the same for any pair $\alpha|\overline{\alpha}$ and any state from the subspace.

\begin{figure}[ht]
\includegraphics[width=0.6\columnwidth]{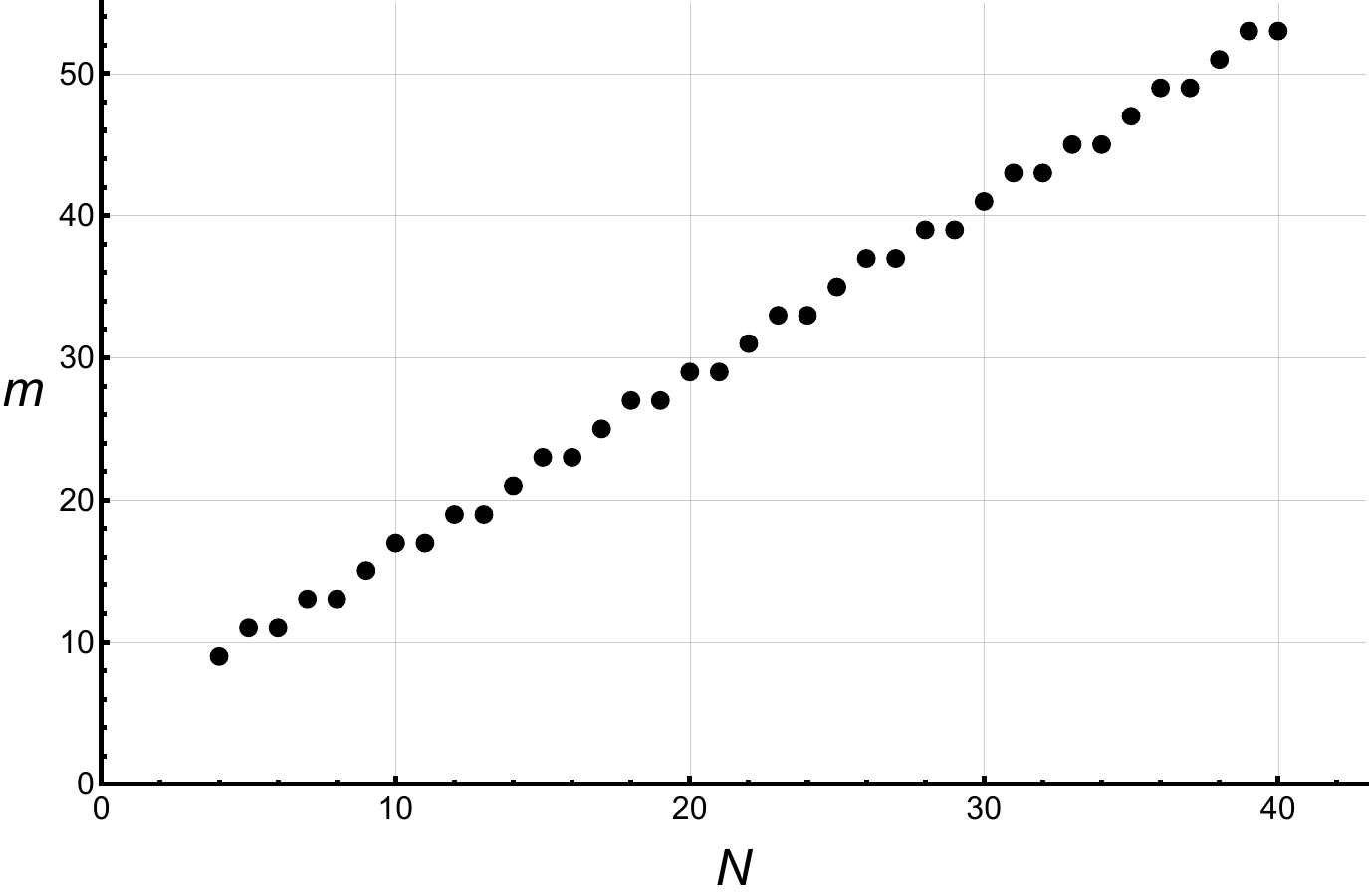}
\caption{The minimal number of measurement settings per party $m$ required for detecting GMNL for $N$ qubit stabilizer subspaces, using chained inequality, Theorem \ref{thm:inequality}, and an assumption that $\tilde{p}_{NL}^{\alpha|\overline{\alpha}}$ is the same for all $\alpha,\overline{\alpha}$.}
\label{fig}
\end{figure}

Furthermore, Ineq. \eqref{ineq:p_NL>chained} can be also used to find the lower bound on $\tilde{p}_{NL}$ as a function of $m$ for a given $N$. In Fig. \ref{fig2} we present this exact relationship for measurement on a state from the five-qubit code \cite{Bennett_1996, PhysRevLett.77.198}.

\begin{figure}[ht]
\includegraphics[width=0.6\columnwidth]{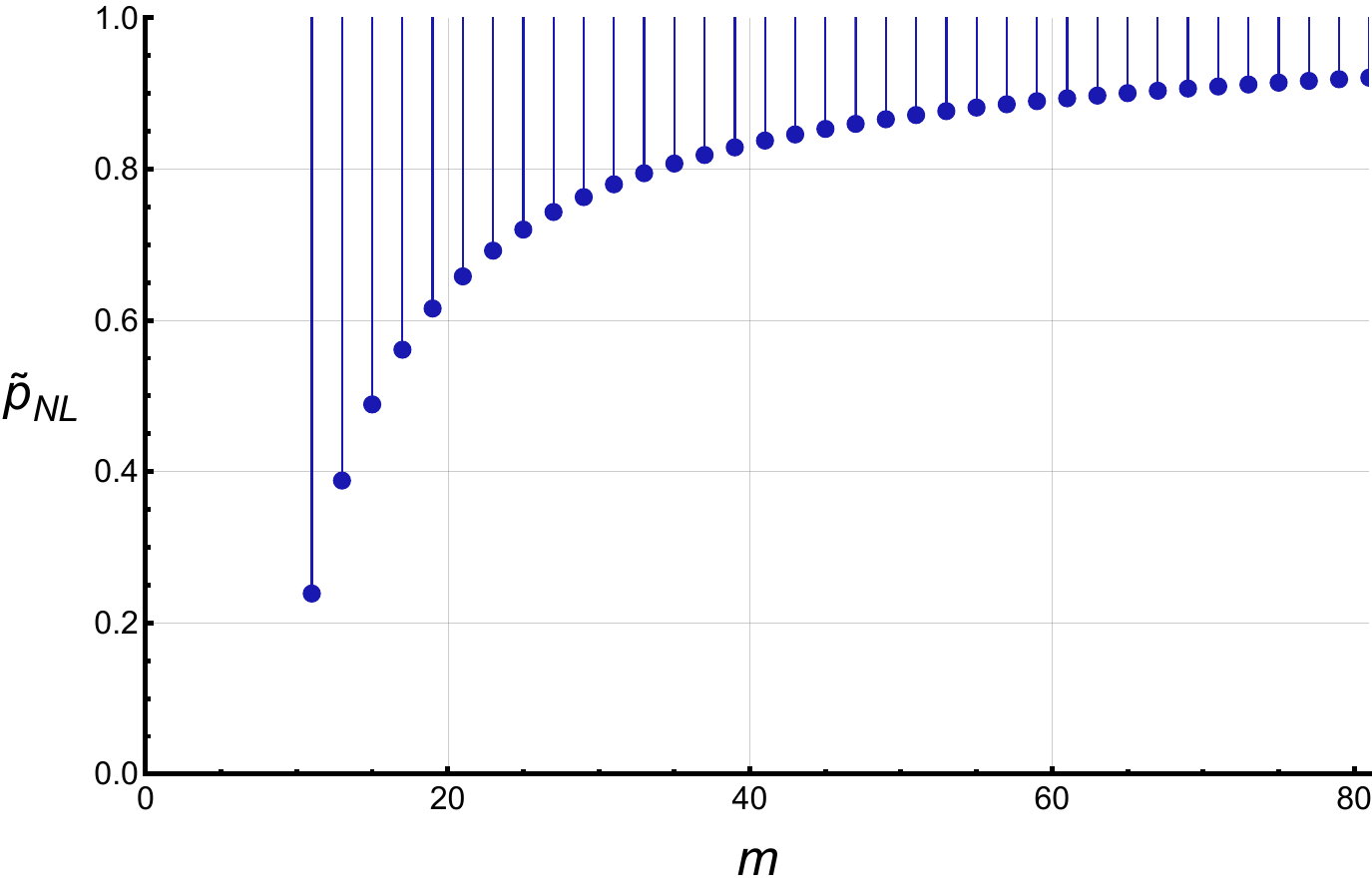}
\caption{The permitted values of $\tilde{p}_{NL}$ by the maximal violation of the chained Bell inequality as a function of the total number of measurements per party $m$, for $N=5$. The blue dots represent the minimal value of $\tilde{p}_{NL}$ that can achieve the maximal violation of the chained Bell inequality for given $m$.}
\label{fig2}
\end{figure}

\end{document}